\documentclass{eptcs}
\providecommand{\event}{EXPRESS 2012} % Name of the event you are submitting to
\def\titlerunning{Approximating Weak Bisimilarity of Basic Parallel Processes}
\def\authorrunning{P. Hofman \& P. Totzke}
\usepackage{version}
\usepackage{etex}
\usepackage{amssymb,amsmath,amscd,amsfonts,latexsym}
\usepackage{tikz}
\usepackage{pgf}
\usetikzlibrary{arrows,automata,backgrounds,decorations}
\usetikzlibrary{shapes.multipart}
\usepgflibrary{decorations.pathreplacing} 
\usepackage{graphicx}
\usepackage{xy}
\usepackage{multirow}
\usepackage{booktabs}
\usepackage{todonotes}
\usetikzlibrary{arrows,calc,topaths}
\tikzstyle{small}=[font=\footnotesize]
\tikzset{
    every picture/.style={>=stealth,auto,node distance=2cm},
}
\RequirePackage[T1]{fontenc}
\RequirePackage{times}

\newcommand{\N}{\mathbb{N}}

\newcommand{\x}{\times}

\newcommand{\R}{Spoiler}
\newcommand{\V}{Duplicator}
\newcommand{\step}[1]{\Step{#1}{}{}}
\newcommand{\wstep}[1]{\Wstep{#1}{}{}}
\newcommand{\Wstep}[3]{\ensuremath{\stackrel{#1}{\Longrightarrow}\stackrel{\scriptstyle{#2}}{\scriptstyle{#3}}}}
\newcommand{\Step}[3]{\ensuremath{\stackrel{#1}{\longrightarrow}\stackrel{\scriptstyle{#2}}{\scriptstyle{#3}}}}
\newcommand{\mbsim}{\sim}

\newcommand{\mwbsim}{\approx}
\newcommand{\Ord}{\ensuremath{Ord}}

\newcommand{\nic}[1]{}

\nic{
}
\newcommand{\anybis}{\approx}
\nic{
\newcommand{\brabis}{\approx_\text{b}}

}

\newcommand{\brabis}{\anybis}

\newcommand{\lra}[1][]{\stackrel{#1}{\longrightarrow}}
\newcommand{\Lra}[1][]{\stackrel{#1}{\Longrightarrow}}

\newcommand{\eps}{\varepsilon}
\newcommand{\Proc}[1]{#1^\otimes}

\newcommand{\Alf}{Act}

\newcommand{\slremove}[1]{}
\newcommand{\phremove}[1]{}
\renewcommand{\xi}{\zeta}
\newenvironment{proof}[1][proof]{\begin{trivlist}
\item[\hskip \labelsep {\itshape #1}]}{\end{trivlist}}

      \newcommand{\qed}{\hfill \mbox{\raggedright \rule{.07in}{.1in}}}

\newtheorem{theorem}{Theorem}[section]
\newtheorem{lemma}[theorem]{Lemma}
\newtheorem{proposition}[theorem]{Proposition}

\newtheorem{definition}[theorem]{Definition}
\newtheorem{example}[theorem]{Example}

\newtheorem{remark}[theorem]{Remark}

\begin{document}

\pagestyle{plain}

\title{Approximating Weak Bisimilarity of Basic Parallel Processes}

\author{ Piotr Hofman \and Patrick Totzke}

%\institute{Institute of Informatics, University of Warsaw \and LFCS, School of Informatics, University of Edinburgh\\}

\maketitle

\begin{abstract}
    This paper explores the well known approximation approach to decide weak bisimilarity of
    \emph{Basic Parallel Processes}. We look into how different refinement functions
    can be used to prove weak bisimilarity decidable for certain subclasses. We also show their
    limitations for the general case. In particular, we show a lower bound of $\omega*\omega$ for
    the approximants which allow weak steps and a lower bound of $\omega+\omega$ for the
    approximants that allow sequences of actions. The former lower bound negatively answers the open question
of Jan\v{c}ar and Hirshfeld.
\end{abstract}

\section{Introduction}
\emph{Basic Parallel Processes} (BPP) were introduced by Christensen \cite{Chr1993} as derivations of commutative context-free
grammars and are equi-expressible with communication-free Petri nets or process algebra using action prefixing, choice
and full merge only.
We are interested in deciding the problem of \emph{weak bisimilarity} for BPP, which remains unresolved even for
normed systems.

Christensen, Hirshfeld and Moller first prove the decidability of strong bisimulation between BPP \cite{CHM1993}, Srba and
Jan\v{c}ar \cite{Srb2002,Jan2003} show the PSPACE completeness of the problem.
For the subclass of normed systems -- where every process has a finite distance to termination -- a polynomial time
algorithm for bisimulation exists \cite{HJM1996a}.
On the negative side, Hirshfeld \cite{Hir1993} proves trace equivalence undecidable for BPP and
H{\"u}ttel \cite{H1994,HKS2009}
shows that indeed all equivalences that lie between strong bisimulation and trace equivalence in the linear/branching
time spectrum \cite{Gla2001} are undecidable.

The main obstacle for deciding weak bisimulation is that one abstracts from silent moves and therefore allows for infinite branching.
Weak bisimilarity is known to be PSPACE-hard for the whole class \cite{Srb2002} and still NP-hard \cite{Str1998} for the subclass of
totally normed systems, which forbids variables of zero and infinite norm.
Stirling \cite{Sti2001} showed that it is decidable for a non-trivial subclass
that still allows infinite branching albeit in a restricted form.  Branching bisimulation for normed
BPP is shown to be decidable in \cite{CHL2011}. However, the technique used there cannot be easily
transferred to work also for weak bisimulation. The problem is that in weak bisimulation games \V\
can go through many equivalence classes when making a move.  This makes it hard to find a connection
between the sizes of \V s configurations before and after move.

Milner originally defines (weak) bisimulation by refinement as the limit of a decreasing sequence
of approximants. This definition is known to coincide with the more customary co-inductive definition
due to Park but the sequence of approximants does not necessarily converge at a finite level for infinitely branching
systems.

We explore the \emph{approximation approach} which is outlined as follows. Weak
bisimilarity is a congruence over a commutative monoid and therefore semi-linear \cite{ES1969},
which means we can enumerate all candidate relations. The fact that the weak bisimulation condition
is expressible
in Presburger Arithmetic means that we can determine for each such candidate if it is a weak bisimulation that contains
a given pair. Hence, a semi-decision procedure for inequivalence immediately implies decidability.
The approximation method discussed here yields such a semi-decision procedure under two assumptions:
1) $\mwbsim$ is finitely approximable: The sequence of approximants stabilizes at level $\omega$, the first limit ordinal.
2) Each approximant $\mwbsim_o$ for $o<\omega$ is decidable.
If both hold true, one can simply iterate through all approximants and for each one check if the given pair of processes
is not contained. The first condition guarantees that this procedure terminates after finitely many rounds for any pair
of inequivalent processes.

Because finite approximation fails for most interesting subclasses we focus on more rigorous refinement functions than the
ones typically considered. We successfully apply the approximation method to restricted classes of BPP: We
derive a decision procedure for checking weak bisimulation for a class defined by Str\'{\i}brn{\'a} in \cite{Str1998}
that allows only a single visible action and no variables of $0$ norm.
Moreover, we provide a new proof for the decidability of weak bisimulation for the class defined by
Stirling \cite{Sti2001}.

We show a lower bound of $\omega * \omega$ for the convergence index of the approximants considered
previously, falsifying a conjecture that is attributed to Hirshfeld and
Jan\v{c}ar\footnote{To our knowledge this conjecture appears in print only in Str\'{\i}brn{\'a}'s PhD
    thesis \cite{Str1998}}
that approximants stabilise at level $\omega + \omega$.
Moreover we show that the most powerful notion of approximation under consideration, for which the individual approximants do not
even need to be decidable themselves are not guaranteed to converge below level $\omega+\omega$.

\section{Preliminaries}
    We write $\Proc{V}$ for the set of all multisets over the finite domain $V$, $\alpha \beta$ for the
    multiset union of $\alpha, \beta\in\Proc{V}$ and $\eps$ for the empty multiset.
    We use $\sqsubseteq$ for multiset (pointwise) inclusion and 
    $\mathcal{P}:V^*\to\Proc{V}$ is the \emph{Parikh} mapping that assigns a word over a finite
    alphabet the multiset that agrees on all multiplicities.
    Write $\Ord$ for the class of ordinal numbers.

    \begin{definition}[Basic Parallel Processes]
        A \emph{process description} is given by a finite set $V =\{X_1, \ldots, X_n\}$ of variables, a finite set
        $\Alf$ of actions and a finite set $T$ of transition rules of the form $X\step{a} \alpha$ where
        $X\in V$, $a\in\Alf$ and $\alpha\in\Proc{V}$.

        A \emph{process} is a multiset in $\Proc{V}$ and may be understood as the parallel composition 
        $X_1^{l_1}\ldots X_n^{l_n}$ of $l_1$ copies of $X_1$, \dots, and $l_n$ copies of $X_n$. 
        %Write $\eps$ for the \emph{empty process}, where $a_1 = \dots = a_n = 0$.
        The behavior of a process is determined by the following extension rule: 
        \[
            \text{if } X\step{a}\alpha\in T \text{ then }X\beta\step{a} \alpha\beta\text{ for any }\beta\in \Proc{V}.
        \]
    \end{definition}
    We assume a dedicated symbol $\tau \in \Alf$ that is used to model \emph{silent} steps $\step{\tau}$
    and define \emph{weak steps} by $\wstep{\tau} = \step{\tau}^*$ and
    $\wstep{a} = \step{\tau}^*\step{a}\step{\tau}^*$ for $a\in\Alf\setminus\{\tau\}$. Weak steps
    are extended to sequences of actions inductively: for the empty word let
    $\wstep{}=\wstep{\tau}=\step{\tau}^*$, for non-empty sequences define $\wstep{aw}=\wstep{a}\wstep{w}$ for $a\in\Alf, w\in\Alf^*$.
        %we write $\alpha \wstep{} \beta$ if process $\beta$ can be reached from $\alpha$ by possibly empty sequence of $\lra[\tau]$ steps.
        %For a given word $w=a_1a_2\dots a_k\in \Alf^*$ we write $\alpha \wstep{w} \beta$ if there is a path
        %of the form $\alpha\Lra\lra[a_1]\Lra\ldots\Lra\lra[a_k]\Lra\beta$ and say $\alpha$ makes a weak $w$-step to
        %$\beta$. Additionaly for $w\in \{\tau\}^*$ let $\wstep{w} = \wstep{}$.
        A \emph{deadlock} is a process that cannot make any non-silent steps.
        The \emph{norm} $|\alpha|$ of a process $\alpha$ is length of the shortest word $w\in \Alf^*$ such that
        $\alpha\wstep{w}\delta$ for a deadlock $\delta$ and $\infty$ if no such sequence exists.
        We call a system \emph{normed} if all its variables have finite norm.
%How about this:
%inductively define:
%=tau=>  = -tau->*
%=a=>    = -tau->*-a->-tau->*  for a\in (Act-tau)
%=aw=>   = =a=>=w=> for a\in Act, w\in Act*
%then we have defined =w=> for all possible w\in Act* consistently,

    \begin{definition}[Weak Bisimilarity]\label{def:wbis}
    A symmetric binary relation $B$ over processes is a \emph{weak bisimulation} iff every pair $\alpha B\beta$ and
    $a\in \Alf^*$ satisfies: if $\alpha\lra[a]\alpha'$ then $\beta\Lra[a]\beta'$ such that
    $\alpha' B \beta'$.
    Two processes $\alpha$ and $\beta$ are weakly bisimilar, denoted 
    $\alpha\brabis\beta$, 
     if there exists a weak bisimulation $B$ such that $\alpha B \beta$.
    \end{definition}

    Following \cite{Mil1989} we characterize weak bisimilarity inductively by refinement:
    
    \begin{definition}[Approximants]\label{def:approximants_orig}
        For a given monotone refinement function $\Psi:2^{\Proc{V}\x\Proc{V}}\to 2^{\Proc{V}\x\Proc{V}}$ we define a decreasing
        sequence of \emph{approximants}, subsets of $\Proc{V}\x\Proc{V}$ by transfinite induction:
        \begin{itemize}
            \item $\mwbsim_0 = \Proc{V}\x\Proc{V}$
            \item $\mwbsim_{i+1} = \Psi(\mwbsim_i)$ for successor ordinals $i+1$ and 
            \item $\mwbsim_\lambda = \bigcap_{i<\lambda} \mwbsim_i$ for limit ordinals $\lambda$
        \end{itemize}
    \emph{Weak Bisimulation approximants} are those based on the refinement function $\mathcal{F}$ that maps
    any $R\subseteq \Proc{V}\x\Proc{V}$ to the largest symmetric relation that satisfies for all $a\in \Alf$ and $\alpha'\in \Proc{V}$:
    \begin{align*}
        (\alpha,\beta)\in \mathcal{F}(R) \iff \alpha\step{a}\alpha'\ \text{ implies }\ \exists \beta'. \beta\wstep{a}\beta' \land (\alpha',\beta')\in R.
    \end{align*}
    \end{definition}

    Every post-fixpoint\footnote{an element $R\subseteq \Proc{V}\x\Proc{V}$ that satisfies $R\subseteq\mathcal{F}(R)$.}
    of $\mathcal{F}$ is a weak bisimulation and by a
    straightforward application of a fixpoint theorem due to Knaster and Tarski we see that the
    sequence of approximants defined by $\mathcal{F}$ converges to weak bisimilarity: $\mwbsim = \bigcap_{o\in \Ord}\mwbsim_o$.
    Thus, if we have a pair of inequivalent processes $\alpha,\beta$, then there is a least ordinal $c$ such that
    $\alpha\not\mwbsim_c\beta$. See \cite{Mil1989}, sec 4.6 for a more detailed account.
    Let the \emph{convergence index} for a class of processes be the least ordinal $c$ such
    that $\mwbsim = \mwbsim_c$ for any system of that class.
    %We say that weak bisimulation
    %is \emph{finitely approximable} for a class of systems if the convergence index is less or equal $\omega$, the first
    %limit ordinal.

    Weak bisimilarity can be characterized in terms of interactive games
    between two players, sometimes called \R\ and \V\ \cite{Sti1998}. For a given pair of processes $\alpha$ and $\beta$, the game 
    consists of a series of rounds. In each round \R\ chooses left or right process and performs a step from it, next \V\ must match this
    with an equally labeled weak step in the other process. If one of the players is not able to perform his next move
    then his opponent wins, infinite plays are won by \V.
    
    \begin{proposition}\label{prop:games}
    Two processes are weakly bisimilar iff \V\ has a strategy to win the bisimulation game regardless of his opponents
    choices.
    \end{proposition}

    In the same spirit we can define \emph{approximants games} % for ordinal level $i\in\Ord$
    to characterize weak bisimulation approximants. % $\mwbsim_i$.
    A configuration of the game consist of a number $o\in Ord$ and a pair of processes $\alpha$ and $\beta$.
    In each round \R\ chooses a new number $o'\in Ord$ such that $0\le o'<o$ and performs a step to
    $\alpha'$ from one of the processes. Then \V\ responds by an equally labeled weak step from the other process
    to some process $\beta'$. The game continues to the next round which starts from configuration $o', \alpha', \beta'$.
    If one of the players is not able to perform his next move then his opponent wins.  This game
    cannot continue indefinitely because $Ord$ is well-founded.

    \begin{proposition}\label{prop:app-games}
    For any $o\in\Ord$ %and processes $\alpha$ and $\beta$ we have
    $\alpha\mwbsim_o \beta$ iff $\V$
    %wins in approximants game which start from $o, \alpha, \beta.$
    has a strategy to win the approximant game from $(o, \alpha, \beta)$ regardless of his opponents choices.
    %If $c$ is a limit then \V\ chooses a next lower distinguishing ordinal when making his response.
    \end{proposition}
    The intuition is that whenever \R\ makes his move to $o', \alpha'$ he asserts that he can win
    the bisimulation game in fewer than $o'$ rounds from the next round onwards, for any possible
    response of his opponent.
    \V\ wins the approximants game at some limit ordinal level only if for all smaller ordinals
    $o'$ he has some response that allows him to win at level $o'$.
    If in the following we write \R\ can \emph{distinguish} processes $\alpha$ and $\beta$ in $o$ rounds we
    mean that \R\ wins the approximant game from $(o,\alpha,\beta)$.

    \begin{example}\label{example:gg}
        Consider the process description given below, where the left-hand side
        is a graphical depiction of the rules listed to the right.
        The left shows a loop $Y \step{aA} Y$ whenever there is a rule $Y\step{a}YA$ in the
        process definition on the right-hand side.
        \vspace{0.5cm}

        {
        \centering
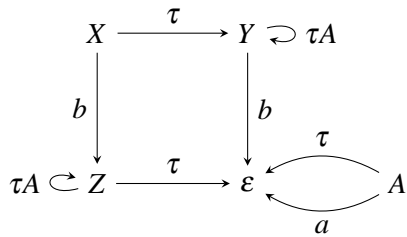
            \begin{minipage}[c]{0.5\textwidth}
                \begin{tikzpicture}
                    \node (X) {$X$};
                    \node (Z) [below of=X] {$Z$};
                    \node (Y) [right of=X] {$Y$};
                    \node (e) [below of=Y] {$\varepsilon$};
                    \node (A) [right of=e] {$A$};
                    \path[->] (X) edge node[left] {$b$} (Z);
                    \path[->] (Y) edge node[right] {$b$} (e);
                    \path[->] (X) edge node {$\tau$} (Y);
                    \path[->] (Z) edge node {$\tau$} (e);
                    \path[->,bend right] (A) edge[above] node {$\tau$} (e);
                    \path[->,bend left] (A) edge node {$a$} (e);
                    \path[->] (Y) edge [loop right] node {$\tau A$} (Y);
                    \path[->] (Z) edge [loop left] node {$\tau A$} (Z);
               \end{tikzpicture}
            \end{minipage}
            \begin{minipage}[c]{.5\textwidth}
                \begin{align*}
                    X \step{\tau} &Y,
                    X \step{b} Z, \\
                    Y \step{b} &\varepsilon,
                    Y \step{\tau} YA, \\
                    Z \step{\tau} &\varepsilon,
                    Z \step{\tau} ZA, \\
                    A \step{\tau} &\varepsilon,
                    A \step{a} \varepsilon
                \end{align*}
            \end{minipage}
        }
    \end{example}
    The two processes $X$ and $Y$ are inequivalent, \R\ wins the bisimulation game by playing
    $(X\step{b}Z)$; any proper response is to $A^n$ for some $n$. Now \R\ continues to play
    $(Z\step{\tau}AZ\step{a}Z)$ $n$ times and wins in the next round. 
    Still, \V\ wins the approximant game from $(\omega,X,Y)$ because
    $ZA^i\mwbsim_jA^j$ for any two naturals $i,j$ and any \R\ attack to some $j,ZA^i$ in the fist round can be
    replied to by a weak step $Y\wstep{\tau^j}YA^j\step{b}A^j$. 
    Hence $\mwbsim \neq \mwbsim_\omega$.

    Example 1 shows that for the usual notion of approximants, the convergence index is above $\omega$, so the
    approximation method fails.  We will continue to investigate different refinement functions that yield faster
    converging weak bisimulation approximants.

\section{Approximants}
    Proposition \ref{prop:app-games} motivates the definition of alternative refinement functions
    and thus approximants by changing the rules of the approximants game. That is, we define
    sequences of faster converging approximants by describing the abilities of the two players to
    move in one round of the game.

    \begin{definition}\label{def:approximants}
        We define different approximants by describing the way both players are allowed to move
        during the approximants game. In all cases \R\ chooses the next lower ordinal and moves
        to some configuration, then \V\ moves from the other process.

        Define ordinary \emph{short-long} approximants $\mwbsim_i$ by the game in which
        \R\ moves along a strong step $\step{a}$, then \V\ responds using a weak step $\wstep{a}$.

        For \emph{long-long} approximants $\mwbsim^L_i$, \R\ makes a weak step
        $\wstep{a}$, then \V\ responds with a weak step $\wstep{a}$.

        For \emph{word} approximants $\mwbsim^W_i$, \R\ moves according to a \emph{sequence}
        $\wstep{w}$ of weak steps where $w\in\Alf^*$, then \V\ responds by a move $\wstep{w}$ over
        the same word.

        \emph{Parikh} approximants $\mwbsim^P_i$ are due the game where \R\ 
        makes a sequence of weak steps $\wstep{w}, w\in\Alf^*$, then \V\ responds
        by a sequence $\wstep{w'}$ in which the letters of $w$ are arbitrarily shuffled: $\mathcal{P}(w)=\mathcal{P}(w')$.
    \end{definition}
    Note that the short-long approximants defined here are exactly the ones given in Definition
    \ref{def:approximants_orig} and all others should converge faster as they give more power to \R.
    We continue to show that all four types of approximants are indeed correct notions of approximation
    for weak bisimilarity and do not converge towards something even smaller in the limit.
    Afterwards, we look at how suitable they are for the approximation method we have in mind.

    \begin{lemma}\label{lemma:approximants}
        For any ordinal $i$, 
        $\mwbsim\ \subseteq\ \mwbsim_i^{W}\ \subseteq\ \mwbsim_i^{P}\ \subseteq\ \mwbsim_i^L\ \subseteq\ \mwbsim_i$.
    \end{lemma}

    \begin{proof}
        For the first inclusion
        %assume $(\alpha,\beta)\notin\ \mwbsim^W_i$, so \R\ can distinguish $\alpha$ and $\beta$
        %in $i$ rounds of a game where in each round he may make a sequence of weak steps from either side and his
        %opponent must respond with a move over the same sequence. Now
        assume that $(\alpha,\beta)$ is in $\mwbsim$, so
        there is a weak bisimulation $B$ containing this pair. This means for any move
        $\alpha_0\step{a_1}\alpha_1\step{a_2}\dots\step{a_k}\alpha_k, a_j\in\Alf$ there is a sequence
        $\beta_0\wstep{a_1}\beta_1\wstep{a_2}\dots\wstep{a_k}\beta_k$ with $\alpha_j B \beta_j$ for $j\le k$,
        so $B$ prescribes a winning strategy for \V\ in the word-approximant game.

        For the second inclusion observe that if \V\ has a response $\beta\wstep{w}\beta'$ for some attack
        $\alpha\wstep{w}\alpha'$ clearly the same response is allowed in the game where he may arbitrarily shuffle the
        letters of $w$.
      
        For the third inclusion assume $(\alpha,\beta)\not\in\ \mwbsim_i^L$, then \R\ can distinguish the two processes
        in $i$ rounds where he only uses weak steps $\wstep{a}$ labelled by single actions and his opponent may also respond
        using equally labelled weak steps. But the same strategy will be winning for \R\ if he is allowed to make
        steps $\wstep{w}$ due to sequences of actions and his opponent may arbitrarily shuffle the actions in his
        response: If an attack is due to a single action the response must be due to a single action. Thus \R\ can
        distinguish $(\alpha,\beta)$ in at most $i$ rounds of this Parikh-game: $(\alpha,\beta)\not\in\ \mwbsim_i^P$.
        The last inclusion follows similarly: If \R\ can distinguish two processes in $i$ rounds if he is only allowed to
        make strong steps $\step{a}$ and his opponent can do weak steps as response, then must also be able to
        distinguish the processes in at most $i$ rounds of a game in which he can also make weak attacks.
    \qed
\end{proof}

    \begin{theorem}\label{thm:convergence}
        $\mwbsim = \bigcap_{i\in \Ord}\mwbsim_i^W = \bigcap_{i\in \Ord}\mwbsim_i^P = \bigcap_{i\in \Ord}\mwbsim_i^L = \bigcap_{i\in \Ord}\mwbsim_i$
    \end{theorem}
    \begin{proof}
        The chain of inclusions $\subseteq$ holds by transfinite induction using Lemma \ref{lemma:approximants}.
        Milner \cite{Mil1989} shows that sequence of short-long approximants converges to weak bisimilarity:
        $\mwbsim = \bigcap_{i\in \Ord}\mwbsim_i$.
    \qed
\end{proof}

    \begin{lemma}\label{lem:congruence}
        For any BPP description and ordinal $i$ we have
        \begin{enumerate}
            \item $\mwbsim_i^L,\mwbsim_i^P,\mwbsim_i^W$ are equivalences and% $\mwbsim_i$ is not in general.
            \item for $\mbsim\ \in\{\mwbsim_i,\mwbsim_i^L,\mwbsim_i^P,\mwbsim_i^W\}$ it holds that
                $\alpha\mbsim\beta$ implies $\alpha\gamma\mbsim\beta\gamma$.
        \end{enumerate}
    \end{lemma}
    \begin{proof}
        1. Let $O \in \{L,P,W\}$. We show transitivity by induction: $\mwbsim^O_0=\Proc{V}\x\Proc{V}$ is trivially
        transitive. Assume $\mwbsim^O_i$ is transitive for $i\in\Ord$ and 1) $\alpha\mwbsim^O_{i+1}\beta$ and 2)
        $\beta\mwbsim^O_{i+1}\gamma$.
        We show that \V\ wins the $O$-approximants game that starts at $(i+1,\alpha,\gamma)$.
        %\V\ wins the approximants game for the pair $\alpha\mwbsim^O_{i+1}\gamma$ by playing a "copycat" strategy:
        Without loss of generality one can assume that \R\ moves $\alpha\wstep{u}\alpha'$. By 1) we
        know that in the game $\alpha$ vs.\ $\beta$ there is a valid response
        $\beta\wstep{v}\beta'$ such that $\alpha'\mwbsim^O_i\beta'$. Equally well if in the game
        $\beta$ vs.\ $\gamma$, \R\ moves $\beta\wstep{v}\beta'$ then by 2) there is a valid
        response $\gamma\wstep{w}\gamma'$ with $\beta'\mwbsim^O_i\gamma'$. By induction hypotheses
        we have $\alpha\mwbsim^O_i\gamma$, so by definition of $\mwbsim^O_{i+1}$ also
        $\alpha\mwbsim^O_{i+1}\gamma$. 

	For limit ordinals $l$ this goes analogously: for \R s
        attack from $\alpha$ there is a response from $\beta$ for all smaller ordinals $i$; for any such move there is a
        response from $\gamma$ to some process equivalent at level $i$. By assumption $\alpha\mwbsim^O_i\gamma$ and
        hence $\alpha\mwbsim^O_l\gamma$ by definition.
        Symmetry and reflexivity follow trivially from the definition.

        The second claim is a result of \V\ using a strategy that remembers which parts of the
        configurations $\alpha\gamma, \beta\gamma$ come from $\alpha, \beta$ and $\gamma$. Every
        move of \R\ from $\alpha\gamma$ (or $\beta\gamma$) can be split into two parts, one which originates 
        from $\alpha$ (or $\beta$) and the one which was performed from
        variables that come from $\gamma$. \V s response will be the combined responses for the first
        and the second part of \R s attack in the games $\alpha$ vs.\ $\beta$ and $\gamma$ vs. $\gamma$.
        In the second part \V\ simply copies \R s move and can therefore even preserve equality
        on the parts of the processes that derive from $\gamma$.
        This means \R\ cannot distinguish $\alpha\gamma$ and $\beta\gamma$ in fewer
        rounds than he can distinguish $\alpha$ and $\beta$.
        \qed
\end{proof}

    The first claim of the lemma does not hold for the short-long approximants $\mwbsim_i$  because \R\ and \V\ have
    different abilities to move. For a counter-example to their transitivity consider example below.
     \begin{example} The following rules describe a system with $X\mwbsim_1 Y \mwbsim_1 Z \not\mwbsim_1 X$:
        $$Y\step{\tau} X,\ 
        Y  \step{\tau} Z,\ 
        Y  \step{\tau} Y',\ 
        Y' \step{a} \varepsilon,\ 
        Y' \step{b} \varepsilon,\ 
        X  \step{a} X,\ 
        Z  \step{b} Z.$$
    \end{example}

    We will continue to show that for finite ordinals $i<\omega$, the approximants
    $\mwbsim_i,\mwbsim^L_i$ and $\mwbsim^P_i$ are decidable.
    For this we recall \emph{Presburger Arithmetic}, the first order logic of natural numbers with addition and
    equality. Syntactically, a Presburger Arithmetic formula is $True, False$, 
    a statement $t_1=t_2$ where the terms $t_1,t_2$ are sums of natural numbers or variables, any
    boolean combination of smaller formulae or a universally or existentially quantified formula.
    We write $F(x_1,x_2\dots x_k)$ for the formula $F$ in which the variables $x_1\dots x_k$ occur
    freely, i.e. not in the scope of a quantifier and interpret formulae over natural numbers and equality.
    A set $R\subseteq\N^k$ of $k$-tuples of natural numbers is said to be \emph{Presburger-definable} if there is a
    Presburger Arithmetic formula $\Phi_R(x_1x_2\dots x_k)$ that satisfies
    $$\Phi_R(x_1x_2\dots x_k)\equiv True \iff (x_1x_2\dots x_k)\in R$$
    
    An important property of Presburger Arithmetic is that it is decidable if a given a Presburger Arithmetic formula $\Phi$ without free variables
    is True. This implies that Presburger-definable sets are decidable.
    Moreover, the class of Presburger-definable sets coincides with the class of \emph{semi-linear}
    sets \cite{GS1966} which for our purposes means it is effectively closed under projection and intersection.
    We refer to \cite{GS1966} for the details on Presburger Arithmetic.

    Any relation $R$ over BP processes with $k$ variables is a subset of $\N^{2k}$.
    %In order to treat multisets $\mu\in\Proc{D}$ as tuples in $\N^D$ we assume some order on its elements
    %and write $i(a)$ for the position of $a\in D$ in that order.
    We now show that for finite $n$, the approximants $\mwbsim_n,\mwbsim_n^L$
    and $\mwbsim^P_n$ are effectively Presburger-definable and therefore decidable relations.
    We recall an important result from \cite{Esp1997}, Thm 3.3:

    %We write $\mathcal{P}(w)\in\Proc{\Alf}$ for the \emph{Parikh image} of word $w\in\Alf^*$:
    %the multiset over $\Alf$ that agrees with $w$ on the multiplicities of its contained symbols
    \begin{lemma}
        \label{lem:reach}
        For any BPP description, the set $Reach\subseteq \Proc{V}\x\Proc{\Alf}\x\Proc{V}$
        of triples $(\alpha,\mu,\beta)$ such that $\alpha\step{a_1}\alpha_1\step{a_2}\alpha_2\ldots\step{a_n}\beta$ for some 
        sequence $a_1a_2\ldots a_n\in\Alf^*$ with $\mathcal{P}(a_1a_2\ldots a_n)=\mu$
        is effectively Presburger-definable.
    \end{lemma}
    From this we can conclude that the step and weak step relations $\step{a},\wstep{a}$
    are effectively Presburger-definable:
    The sets $S_1=\{a\}$, and $S_2=\{a\}\Proc{\{\tau\}}$ (in other words the Parikh images of $a\tau^*$) are easily seen to be Presburger-definable
    and $\step{a}$ and $\wstep{a}$ are expressible as the projections into the first and third component of
    $Reach\cap(\Proc{V}\x S_1\x \Proc{V})$ and $Reach\cap(\Proc{V}\x S_2\x\Proc{V})$ respectively.
    
    \begin{theorem}
        \label{thm:decidable-finite-apps}
     For a given $BP$ process description $B$ with $k$ variables the $n$-th approximants
     $\mwbsim_n,\mwbsim_n^L$ and $\mwbsim_n^P$ over $B$ are decidable for all finite $n$.
    \end{theorem}
    \begin{proof}
        It suffices to to show that $\mwbsim_n,\mwbsim_n^L$ and $\mwbsim_n^P$ are effectively
        Presburger-definable. By Lemma \ref{lem:reach} we can assume a Presburger Arithmetic formula $R\subseteq\N^V\x\N^{\Alf}\x\N^V$
        that expresses the set $Reach$ and formulae $Step_a,WStep_a\subseteq\N^V\x\N^{\Alf}\x\N^V$
        expressing the strong and weak $a$-step relations for all actions $a\in\Alf$.
        Now we can easily encode the refinement functions used in the approximants 
        and for any finite $n$ construct the Presburger Arithmetic formulae that express $\mwbsim_n,\mwbsim_n^L$ and
        $\mwbsim_n^P$ by induction:

        For $n=0$ we have $\mwbsim_0=\mwbsim_0^L=\mwbsim^P_0=\N^{2k}$ trivially definable as $\Psi_0(\alpha,\beta) = True$.
        
For
        $\mwbsim_{i+1}$ let $\Psi_{i+1}(\alpha,\beta) \iff \bigwedge_{a\in\Alf}\ ($
        \begin{align*}
            &(\forall\alpha'\in\N^{V}\  Step_a(\alpha,  \alpha') \implies  \exists\beta'\in\N^{V}\ 
                    WStep_a(\beta,  \beta') \land \Psi_i(\alpha',\beta') )\\
            \land
            &(\forall\beta'\in\N^{V}\  Step_a(\beta,  \beta') \implies  \exists\alpha'\in\N^{V}\ 
                    WStep_a(\alpha,  \alpha') \land \Psi_i(\alpha',\beta') ))&
        \end{align*}
        Similarly, for $\mwbsim^L_{i+1}$ let $\Psi_{i+1}(\alpha,\beta)$ as above but replace $Step_a$
        by $WStep_a$.
        For $\mwbsim^P_{i+1}$ let $\Psi_{i+1}(\alpha,\beta) \iff \forall \mu\in \N^{\Alf}\ ($
        \begin{align*}
          &(\forall\alpha'\in\N^{V}\ R(\alpha, \mu, \alpha') \implies  \exists\beta'\in\N^{V}\ 
                    R(\beta, \mu, \beta') \land \Psi_i(\alpha',\beta') )&\\
                    \land
                    &(\forall\beta'\in\N^{V}\  R(\beta, \mu, \beta') \implies  \exists\alpha'\in\N^{V}\ 
                    R(\alpha, \mu, \alpha') \land \Psi_i(\alpha',\beta') ))
        \end{align*}
        \qed
\end{proof}
    It is worth mentioning that word approximants $\mwbsim^W_n$ are not decidable at finite levels: for systems
    without silent actions the very first approximant $\mwbsim_1^W$ coincides with \emph{trace equivalence}, which has
    been shown to be undecidable for BPP by Hirshfeld \cite{Hir1993}.

\section{Applications}
    We now use the approximation approach to show that two subclasses of BPP previously known in the
    literature have decidable weak bisimilarity. In particular, we show this result in Section
    \ref{sec:app:sti} for the class introduced in \cite{Sti2001} by proving weak bisimilarity
    finitely approximable for the long-long approximants $\mwbsim^L$ and for the subclass introduced
    in \cite{Str1998} we show finite approximability for Parikh approximants $\mwbsim^P$ in Section
    \ref{sec:app:str}. In both cases we know by Theorem \ref{thm:decidable-finite-apps} that at
    finite levels the approximants are decidable equivalences and hence showing their convergence at
    level $\omega$ suffices to get a decision procedure.

    \begin{proposition}\label{middle_lemma}
     The following states some useful facts that are easily verified.
        \begin{enumerate}
            \item $\alpha\mwbsim\beta$ implies $|\alpha| = |\beta|$
            \item If $\alpha\wstep{}\beta\wstep{}\alpha'$ and $\alpha\mwbsim\alpha'$ then $\alpha\mwbsim\beta$.
            \item If $\alpha\wstep{}\alpha\beta$ and $\beta$ has norm $0$ then $\alpha\mwbsim\alpha\beta$
        \end{enumerate}
    \end{proposition}

    \begin{definition}
    Let $O\in \{L,P,W\}$ and $\alpha,\beta\in \Proc{V}$ such that $\alpha\mwbsim_{\omega}^O \beta.$ For a given \R\ move
    from $\alpha$ to $\alpha'$ there is a sequence $B=\beta_1',\beta_2',\beta_3'\ldots$ of \V\ responses such that for
    all $i\in N$ holds $\alpha' \mwbsim_{i}^O \beta_i'.$ We call $B$ a \emph{family of responses}.
    \end{definition}
    Observe that the sequence is not unique, for example if you substitute $\beta_i$ by $\beta_j$ for any $j>i$ then you obtain another family of responses.  
    %Moreover if $\alpha \not\mwbsim_{\omega+1}^O \beta$ then a sequence of responses can not stabilize.
    By Dickson's Lemma we can assume that a family of responses is non-decreasing with respect to
    multiset inclusion: $\beta_i\sqsubseteq\beta_{i+1}$ for every $i\in\N$.
    
\subsection{Normed Processes with Pure Generators}
\label{sec:app:sti}
    Write $\alpha\step{}_0\beta$ for silent and norm-preserving steps between processes
    $\alpha,\beta\in\Proc{V}$: $\alpha\step{}_0\beta$ iff $\alpha\step{\tau}\beta$ and
    $|\alpha|=|\beta|$. Let $\wstep{}_0$ be the transitive and reflexive closure of $\step{}_0$.
    For variables $X,Y$ such that $X\Lra_0 Y\Lra_0 X$ we have $X\mwbsim Y$ by Claim 2) Proposition \ref{middle_lemma}. We say
    $X$ is \emph{redundant} because of $Y$ or vice versa. One can easily detect redundant variables and therefore we can
    assume that they have already been unified. That is, we can assume wlog. that our process description does not
    contain redundant variables.
    This allows us to linearly order the set $V$ of variables such that if $X\wstep{}_0Y\alpha$ then $X>Y$. Let's fix the
    notation $X_1>X_2>\ldots >X_k$.

    A \emph{generator} is a variable $X$ that allows a sequence $X\wstep{}_0X\alpha$ for some $\alpha\in\Proc{V}$, in which case
    we say $X$ \emph{generates} $\alpha$. Call a generator $X$ \emph{pure} if $X\wstep{}_0\alpha$ implies that
    $\alpha=\alpha'X$: Pure generators cannot vanish silently.

    Stirling shows decidability of weak bisimilarity for normed processes with only pure generators
    using a tableaux approach \cite{Sti2001}. One motivation for this subclass is
    that it still allows for infinite branching and that ordinary ($\mwbsim_i$) approximants do not
    converge at level $\omega$. In this section we show that long-long ($\mwbsim^L_i$) approximants
    in fact stabilize at level $\omega$ and thus provide the missing negative semidecision prodecure
    to conclude decidability.

    \begin{lemma}\label{lem_finitely_classes}
        Let $\alpha$ be a normed process of a BPP description without redundant variables in which every generator is pure.
        $Succ=\{\alpha'| \alpha\wstep{}_0\alpha'\}$ can be partitioned into finitely many
        equivalence classes with respect to weak bisimilarity.
    \end{lemma}
    \begin{proof}
        The third claim of Proposition \ref{middle_lemma} allows us to restrict
        ourselves to the subset $Succ'$ of $Succ$ of configurations which are obtained without use of generating moves
        because it has the same number equivalence classes as generators cannot vanish along
        $\wstep{}_0$ moves. Our goal is to show that $Succ'$ is finite which immediately implies the
        claim of the lemma. 
      
        Every derivation of $\alpha$ is a sum of derivations from variables belonging to $\alpha$.
        If we prove that in silent norm preserving steps without generating moves, we can only derive finitely
	many configurations from each
        variable, then we will also prove that $Succ'$ is finite.
        We will show that this is indeed the case for all variables by induction over the assumed order $<$.
        From the smallest variable $X_k$ using silent norm
        preserving steps \emph{without generating} we can derive only two configurations, namely $X_k$ or $\eps$.

        Assume $c>0$ bounds the number of possible silent norm preserving derivations from any
        variable in $X_i \ldots X_k$ and consider the variable $X_{i-1}$.
        In case $X_{i-1}$ is a deadlock variable, i.e. $X_{i-1}\step{\tau}X_{i-1}$ is the only applicable
        rule, we can trivially bound the number of its derivations by $1\le c$.
        Otherwise, because we forbid generating moves we must have that any rule $X_{i-1}\step{\tau}_0\alpha$ produces a multiset 
        $\alpha\in \Proc{\{X_i\ldots X_k\}}$.
        The fact that there are only finitely many rules that rewrite variable 
        $X_{i-1}$ implies that we can bound the number of its silent norm preserving derivations by
	    $$d\ \cdot\ c^l+1,$$
        where $d$ is the number of rules for $X_{i-1}$ and $l$ is the maximal size of any right hand
        side of a rule rewriting $X_{i-1}$.
    \qed
\end{proof}

    \begin{theorem}\label{thm:Sti2001}
        $\mwbsim\ =\ \mwbsim^L_\omega$ for normed BPP where each generator is pure.
    \end{theorem}
    \begin{proof}
        Assume towards a contradiction that we have % two processes $\alpha$ and $\beta$ with
        $\alpha\mwbsim^L_\omega\beta\not\mwbsim^L_{\omega+1}\alpha$. Wlog.\ assume an
        optimal\footnote{a move prescribed by an optimal winning strategy: one that guarantees a win for
        Spoiler in the fewest number of rounds and thus properly decreases the approximation index
        in each round.}
        initial move $\alpha\wstep{a}\alpha'$ for \R\ in the game $\alpha$ vs.\ $\beta$ and a family
        $B=\beta'_0,\beta'_1,\ldots$ of responses which is strictly increasing wrt.\ multiset
        inclusion.

        By Lemma \ref{lem_finitely_classes}, the set $Succ = \{\alpha''|\alpha'\wstep{}_0\alpha''\}$
        of configurations reachable from $\alpha'$ in silent and norm-preserving steps
        contains finitely many bisimilarity classes. Let the set $Succ'$ be a finite set of representants of those
        classes in $Succ$.
        This allows us to define a function $f:B\to Succ'$ that maps $\beta_i'\in B$ to an element in $Succ'$ that maximises
        their approximation index: $\beta_i'\mwbsim^L_k f(\beta'_i)$ and $\forall \gamma\in Succ'\  \beta_i'\mwbsim^L_l\gamma \implies k\ge l$.
        This function is well defined because set $Succ'$ is finite.
        Now consider an infinite subsequence $B(\gamma)$ of $B$ that contains all elements which $f$ maps to the
        configuration $\gamma\in Succ'$. By the pigeon hole principle such a subsequence exists.

        Take two different elements $\beta'_i\sqsubset \beta'_j$ of $B(\gamma)$ for arbitrary large
        $i,j$. We have
        1) $\beta'_i \mwbsim^L_i \gamma \mwbsim^L_j\beta'_j$ because $\alpha'\in Succ'$ and 2)
        $\beta'_i$ and $\beta'_j$ have the same norm.
        To see why the second obervation is true note that
        $|\alpha|\neq|\beta|$ implies $\alpha\not\mwbsim^L_{\min{\{|\alpha|,|\beta|\}}}\beta$
        as \R\ only needs to decrease the smaller process to a deadlock which cannot
        be mimiked by \V\ on the other process because the norms differ.
        We know $\beta'_i\mwbsim^L_i\alpha'\mwbsim^L_j\beta'_j$, so $|\beta'_i|=|\alpha'|=|\beta'_j|$
        as otherwise $i$ and $j$ would be bounded by $|\alpha'|$.

        Consider the game on $\alpha'$ vs.\ $\beta'_j$ and a silent, norm-preserving move $\beta'_j\wstep{}_0\beta'_i$
        made by \R, which must be possible due to observation 2) and the fact that $\beta'_i$ is a subset of $\beta'_j$.
        Now by definition of the subsequence $B(\gamma)$ we deduce that $\alpha'\wstep{}_0\gamma$ is
        an optimal response for \V. Therefore by 1), we know that $\beta'_i \mwbsim^L_{j-1}\gamma$
        so $\beta'_i\mwbsim^L_{j-1}\beta'_j$ by transitivity and the fact that
        $\beta'_j \mwbsim^L_{j-1}\gamma$. But now we have $\beta'_i \mwbsim^L_{j-1}\alpha'$ for
        arbitrarily high $j$ and therefore $\beta'_i\mwbsim^L_\omega\alpha'$ which contradicts the
        optimality of \R's very first move.
    \qed
\end{proof}

\subsection{Unnormed Processes over one visible Action}
\label{sec:app:str}
    %TODO: discuss this class: she showed that for this class hirshfelds conjecture holds,
    % the class is essentially totally normed (all positive, finite norm), the infinite norm case is made trivial by the
    % unary assumption.
    Consider the subclass of BPP processes that satisfy both
    \begin{enumerate}
        \item There is only one visible action label, $\Alf = \{\tau, a\}$ and
        \item Every variable has positive or infinite norm.
    \end{enumerate}
    This class has been introduced in \cite{Str1998}, where it was shown that for processes of this kind,
    Hirshfelds conjecture holds: $\mwbsim\ =\ \mwbsim^L_{\omega*2}$. 
    Note that this class is not a subclass of the \emph{totally normed} systems \cite{Hir1997} as it allows
    for variables with infinite norm.  We show that this class has decidable weak bisimilarity by showing that
    Parikh-approximants converge at level $\omega$.

    \begin{theorem}
        $\mwbsim\ =\ \mwbsim^P_\omega$ for the subclass of BPP processes with a single visible
        action and no variables with norm $0$.
    \end{theorem}
    \begin{proof}
        First observe that 1) implies that all configurations with infinite norm must be equivalent and due
        to norm preservation cannot be equivalent to any configuration of finite norm.
        The second restriction guarantees that there are only finitely many different configurations for any given
        finite norm.
        2) Whenever two processes have different but finite norms, they are certainly not related by $\mwbsim^P_2$
        as \R\ may rewrite the smaller process to a deadlock in one long step without allowing his
        opponent to do the same on the other process.

        Assume towards a contradiction that $\alpha\mwbsim^P_\omega\beta\not\mwbsim^P_{\omega+1} \alpha$.
        So for an optimal initial move $\alpha\wstep{w}\alpha'$ for \R\ there is a family of responses from $\beta$.
	This sequence cannot converge as otherwise our
        assumption $\beta\not\mwbsim^P_{\omega+1} \alpha$ would be false. By the pidgin hole principle, there
        must be at least one variable $X$ that grows indefinitely along this sequence. Take two
        elements $\beta'_i\sqsubset\beta'_j,\ 2<i<j$ from this sequence such that $X$ occurs more often in $\beta'_j$.
        By observation 2) and the fact that $\beta'_i$ and $\beta'_j$ have different norms we know
        that $\beta'_i\not\mwbsim^P_2\beta'_j$.
        Because $\beta'_i\mwbsim^P_i\alpha'\mwbsim^P_j\beta'_j$ and $i<j$ holds
        $\beta'_i\mwbsim^P_i\alpha'\mwbsim^P_i\beta'_j$. From this and the transitivity of
        $\mwbsim^P_i $ we conclude that $\beta'_i \mwbsim^P_i\beta'_j$ and because $2<i$ also
        $\beta'_i\mwbsim^P_2\beta'_j$ which is a contradiction.
    \qed
\end{proof}

\section{Limitations of the Approximant Approach}\label{sec:limitations}
One severe limitation of the approximation method is that it cannot provide complexity bounds even if
successfully applied.
In this section we show that $\mwbsim^L$ is not guaranteed to stabilize at level $\omega*2$ and that word approximants $\mwbsim^W$
do not necessarily stabilize on level $\omega$.
From our counter-examples we derive lower bounds of $\omega^2$ and $\omega*2$ for the convergence indices of $\mwbsim^L$
and $\mwbsim^W$ respectively.
%which shows that any ,,easy'' extension of $\mwbsim^L$ will not solve the problem of the bisimulation relation in $BPP$ systems.    

\begin{theorem}\label{thm:hirschfeld}
    Long-Long approximants ($\mwbsim^L_i$) do not stabilize below level $\omega^2$ for BPP:
    $\mwbsim\ \neq\ \mwbsim^L_{\omega * k}\ $ for all finite $k$.
\end{theorem}
    \begin{proof}
        For $k<2$ the claim is trivial, e.g. by Example \ref{example:gg}. We first show how to
        construct a system with $\mwbsim\ \neq\ \mwbsim^L_{\omega +\omega}$. For this we recycle
        Example \ref{example:gg} and add the rule $X\step{\tau}XA$ and analyze the game on
        $X$ vs.\ $Y$ more carefully. The fact that $X$ can be silently rewritten to $Y$ forces \R\
        to start from $X$. Any optimal silent move for \R\ must change the equivalence class, so we
        can assume his initial move to be $X\wstep{b}ZA^m$. 
        %which is equivalent to $Z$ by point 3) of Proposition \ref{middle_lemma}.
        \V\ must respond to some $A^n$. To prevent a perfect match to an identical process in the next round,
        \R\ must again move from $ZA^m$ and may not end in a configuration $A^{<n}$. So \R\ will either move
        $Z\wstep{a}Z$ or $Z\wstep{a}A^m$ with $m\ge n$
        and thereby force \V\ to remove one $A$ on the other side. Observe that any one move
        from $Z$ or $A^m$ can be replied to by $A$, so \R\ has to keep making $a$-moves from his process
        until \V\ has exhausted all variables $A$. By removing only one $A$ in each such response, \V\ can
        prevent the situation $Z$ (or $A^{>0}$) vs.\ $\eps$
        for $n$ rounds, where $n$ is determined by his initial response. We conclude $X\mwbsim^L_\omega Y\not\mwbsim X$.
        
        To construct a counter-example to convergence at level $\omega+\omega$ we combine two copies of this system as
        indicated in Figure \ref{fig:HC-construction} below.

        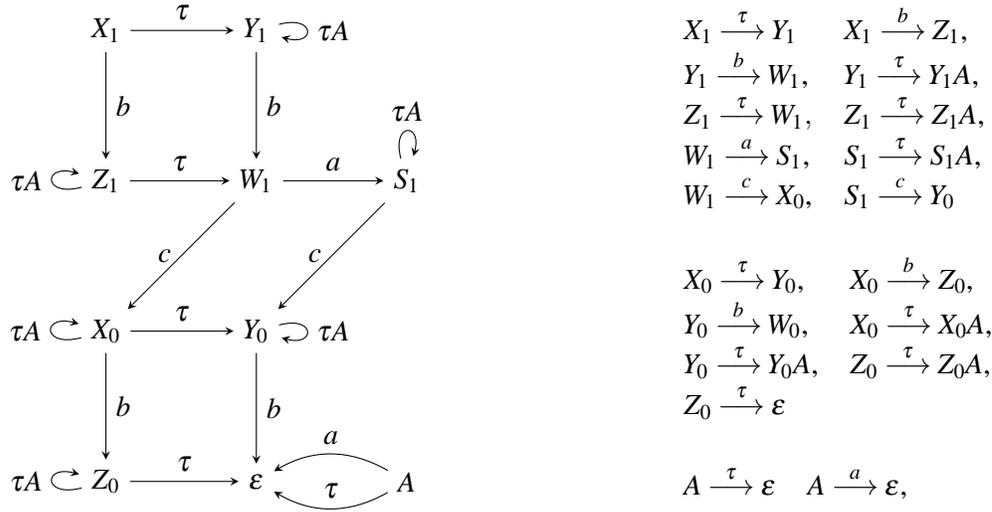
\begin{figure}[h]
            \caption{Combining two copies of the "Guessing Game" yields $X_1\mwbsim^L_{\omega*2}Y_1\not\mwbsim X_1$.}
            \label{fig:HC-construction}
        \begin{center}
            \begin{minipage}[c]{.55\textwidth}
                \begin{tikzpicture}                  
                    \node (X) {$X_1$};
                    \node (Z) [below of=X] {$Z_1$};
                    \node (Y) [right of=X] {$Y_1$}; 
		    \node (W) [below of=Y] {$W_1$};
                    \path[->] (X) edge node {$\tau$} (Y);
                    \path[->] (Z) edge node {$\tau$} (W);
                    \path[->] (X) edge node {$b$} (Z);
                    \path[->] (Y) edge node {$b$} (W);
                    \path[->] (Y) edge [loop right] node {$\tau A$} (Y);
                    \path[->] (Z) edge [loop left] node {$\tau A$} (Z);

		    \node (S0) [right of=W] {$S_1$};
                    \node (X0) [below of=Z] {$X_0$}; 
                    \node (Y0) [right of=X0]{$Y_0$};
                    \node (e) [below of=Y0] {$\eps$};
		    \node (Z0) [below of=X0] {$Z_0$};

                    \path[->] (X0) edge node {$\tau$} (Y0);
                    \path[->] (Z0) edge node {$\tau$} (e);
                    \path[->] (Y0) edge node {$b$} (e);
                    \path[->] (X0) edge node {$b$} (Z0);

                    \path[->] (Y0) edge [loop right] node [right]{$\tau A$} (Y0);
                    \path[->] (X0) edge [loop left] node {$\tau A$} (X0);
                    \path[->] (Z0) edge [loop left] node [left]{$\tau A$} (Z0);

                    \path[->] (W) edge node [left]{$c$} (X0);
                    \path[->] (W) edge node {$a$} (S0);
                    \path[->] (S0) edge node [left]{$c$} (Y0);
                    \path[->] (S0) edge [loop above] node {$\tau A$} (S);

		    \node (A) [right of=e] {$A$};
                    \path[->,bend right] (A) edge node [above] {$a$} (e);
                    \path[->,bend left]  (A) edge node [above] {$\tau$} (e);
               \end{tikzpicture}
            \end{minipage}
            \begin{minipage}[c]{.40\textwidth}
                \begin{tabular}{ll}
		    $X_1 \step{\tau} Y_1$&$X_1 \step{b} Z_1$, \\
                    $Y_1 \step{b} W_1$,& $Y_1 \step{\tau} Y_1 A$,\\ 
                    $Z_1 \step{\tau} W_1, $ & $ Z_1 \step{\tau} Z_1 A$, \\
		    $W_1 \step{a} S_1$, & $S_1 \step{\tau} S_1 A$,\\ 
                    $W_1 \step{c} X_0$, & $S_1 \step{c} Y_0$ 
                \end{tabular}
                \vspace{0.5cm}

                \begin{tabular}{ll}
		    $X_0 \step{\tau} Y_0$, &$X_0 \step{b} Z_0$,\\
                    $Y_0 \step{b} W_0$, &$X_0 \step{\tau} X_0 A$,\\
                    $Y_0 \step{\tau} Y_0 A$, & $ Z_0 \step{\tau} Z_0 A$, \\
		    $Z_0 \step{\tau} \eps$ & \\
                \end{tabular}
                \vspace{0.5cm}

                \begin{tabular}{ll}
                    $A \step{\tau} \eps$ & $ A \step{a} \eps$,\\ 
                \end{tabular}
            \end{minipage}
        \end{center}
        \end{figure}
    The bottom part of the construction is the gadget as discussed previously. Observe that variables
    $X_0,Y_0,Z_0$ are not able to produce variables from the top part of the
    diagram, those variables with an index $1$. Thus we preserve that $X_0 \mwbsim_{\omega}^L Y_0$.   
    Our aim is to show that indeed $X_1\mwbsim^L_{\omega+\omega}Y_1\not\mwbsim X_1$. For this it suffices to show
    that the only possibility for \R\ to win is to force the game from $X_1$ vs.\ $Y_1$ to end up in
    $X_0 \mwbsim_{\omega}^L Y_0$.
        
    The Game starts from a pair $X_1, Y_1$ and it goes through the upper square pattern $X_1, Y_1, Z_1, W_1$. By our previous
    discussion of this gadget, we know that \R\ has to start by $X_1\wstep{b}Z_1A^m$; \V\ will respond to
    $W_1A^n$.
    \R\ must continue to play from the left hand side in order to prevent a perfect match to identical processes
    and cannot move to a $W_1A^i$ for $i\le n$. If he makes a move $Z_1A^m\step{c}X_0A^i$, while the other process
    still contains a $W_1$, \V\ is able to match to the same process.
    So the only option left for \R\ is to force \V\ to remove all variables $A$ one by one by performing $a$-steps.
    %, for example by moving $Z_1\wstep{a} Z_1$. 
    Eventually, from a position $Z_1$ (or $W_1A^{>0}$) vs.\ $W_1$, \R\ makes one last $a$-step and thus forces \V\ to
    rewrite $W_1$ to $S_1$. Afterwards, \R\ can force the game to a position $X_0A^n$ vs.\ $Y_0A^m$ by playing a
    $c$-step from either side. This part of the game takes $n+1$ rounds and $n$ was chosen by \V\ in his first response.
    Therefore $X_1\mwbsim_{\omega+\omega}Y_1$ which completes the proof for $k=2$.

    The construction above can be extended to provide a counter-example for convergence at level $\omega * k$
    for any natural $k$ by stacking $k$ copies of the square gadget on top of each other.
    This can also be modified to a system which contains only variables of the norm zero.
    \qed
\end{proof}

    Next we focus on Word approximants and falsify a conjecture of Str\'{\i}brn{\'a} \cite{Str1998}
    about their convergence above level $\omega$.

    \begin{theorem}\label{thm:non_stabilize}
        %Word-approximants $\mwbsim^W$ are not finitely approximable for BPP:
        %$\mwbsim\ \neq\ \mwbsim^W_\omega$.
        For BPP, weak bisimilarity is not finitely approximable with word approximants: $\mwbsim\ \neq\ \mwbsim^W_\omega$.
    \end{theorem}
    \begin{proof}
        Consider the process description in Figure \ref{fig:Words}.
        \begin{figure}
            \caption{Counter-example for finite approximability of $\mwbsim^W_i$}
            \label{fig:Words}
        \begin{center}
            \begin{minipage}[l]{.6\textwidth}
            \begin{tikzpicture}         
                \node (X) {$X$};
                \node (Z) [below of=X] {$Z$};
                \node (Y) [right of=X] {$Y$}; 
                \node (e) [below of=Y] {$\varepsilon$};
                \node (L) [below of=e] {$L$};
                \node (R) [right of=L] {$R$};
                \node (Q) [left of=L] {$Q$};
                    
                \path[->] (X) edge node[left] {$a$} (Z);
                \path[->] (Y) edge node[right] {$a$} (e);
                \path[->] (X) edge node {$\tau$} (Y);
                \path[->] (Z) edge node {$\tau$} (e);
                \path[->] (Y) edge [loop above] node {$\tau L$} (Y);
                \path[->] (Z) edge [loop left] node {$\tau L$} (Z);

                \path[->] (R) edge [loop below] node {$a$} (R);
                \path[->, bend left] (Y) edge node {$\tau$} (R);
                \path[->] (L) edge node {$\tau$} (e);
                \path[->] (L) edge node[above] {$a$} (R);
                \path[->] (L) edge [loop below] node {$\tau Q$} (L);
                \path[->, bend right] (R) edge node[above] {$\tau$} (e);
	        \path[->, bend left] (Q) edge node {$a$} (e);
	        \path[->] (Q) edge node {$\tau$} (e);
            \end{tikzpicture}
            \end{minipage}
            \begin{minipage}[r]{.35\textwidth}
		\begin{tabular}{ll}
		    $X\step{a}Z$, &$X\step{\tau}Y$\\
                    $Y\step{a}\eps$, &$Y\step{\tau}YL$\\
                    $Y\step{\tau}R$\\
                    $Z\step{\tau}ZL$ , &$Z\step{\tau}\eps$\\

                    $L\step{a}R$, & $L\step{\tau}\eps$\\
                    $L\step{\tau}LQ$\\
                    $R\step{\tau}\eps$ ,& $R\step{a}R$\\
		    $Q\step{a}\eps$, & $Q\step{\tau}\eps$\\
                \end{tabular}
            \end{minipage}
        \end{center}
    \end{figure}
        By Proposition \ref{middle_lemma} part 3, we know that $ZL^nQ^m\mwbsim Z$ and $L^{n+1}Q^m\mwbsim L^{n+1}$
        for any two naturals $m, n$.
        We claim that $X\mwbsim^W_\omega Y\not\mwbsim X$ and base our proof on the following claims that are proven
        individually after the main argument. For $i,j,n\in\N$, $n>0$ we have
        \begin{align}
            &Z \not\mwbsim^W_3 RL^i \not\mwbsim^W_3 L^j\label{claim_1},\\
            &Z \mwbsim^W_{2n+1} L^n\label{claim_2},\\
            &Z \not\mwbsim^W_{2n+2} L^n\label{claim_3}.
        \end{align}
        In the game $X$ vs.\ $Y$ \R\ must start with a move $X\wstep{a}ZL^lQ^q\mwbsim Z$, as otherwise his
        opponent is able to match to the same process and thereby win. Possible responses for \V\ from $Y$ are: 
        \begin{itemize}
            \item To some $RL^nQ^m\mwbsim RL^n$, which allows \R\ to win in 3 further rounds by Claim \ref{claim_1}.
            \item To some $YL^nQ^m\mwbsim YL^n$ which allows \R\ to silently replace the $Y$ by $R$ and afterwards again win in $3$
                rounds by Claim \ref{claim_1}). Note that no silent response from $ZL^lQ^q$ to some
                configuration that contains $R$ is possible.
            \item To some $Q^m$ which allows \R\ to win in one round by playing $Z\wstep{a^{m+1}}Z$.
            \item To some $L^nQ^m\mwbsim L^n, n>0$ which allows \R\ to win but in not fewer than $2n+2$ rounds by Claims \ref{claim_2})
                and \ref{claim_3}).
        \end{itemize}
        The choice of $n$ is made by \V\ and therefore $X\mwbsim^W_\omega Y\not\mwbsim X$.
        Note that this counter-example uses only a single visible action and all variables have zero
        norm.
    \qed
\end{proof}

  It remains to proof claims \ref{claim_1}.-\ref{claim_3}.
  We first prove some auxiliary claims on which we base our arguments for claims 1) and 2). For all $m,n\in\N$,
  \begin{align}
      R L^n &\not\mwbsim^W_1 Q^m \label{claim_4}\\
      L^n &\not\mwbsim^W_2 R\ \not\mwbsim^W_2Z \label{claim_5}
  \end{align}
  For (\ref{claim_4}), observe that \V\ cannot respond to a move $R\wstep{a^{m+1}}R$.
  For (\ref{claim_5}), \R\ moves from $L^n$ (or $Z$) silently to $Q$ and \V\ can respond to $R$
  or to $\eps$. In the first case he loses in one round by claim (\ref{claim_4}), in the latter he
  cannot respond to move $Q\wstep{a}\eps$ from $\eps$.

\paragraph*{Claim (\ref{claim_1}): $Z \not\mwbsim^W_3 RL^n \not\mwbsim^W_3 L^m$.}
\begin{proof}
  For both parts \R\ moves from $RL^n$ silently to $R$. \V\ must respond either to $Q^k$ which is losing for him in
  one round by Claim (\ref{claim_4}), or to $L^jQ^k\mwbsim L^j$ or $Z L^jQ^k\mwbsim Z$, which is losing for him in two
  rounds by Claim (\ref{claim_5}).
\qed
\end{proof}

\paragraph*{Claim (\ref{claim_2}): $Z\mwbsim^W_{2n+1} L^n$ for $n>0$.}
\begin{proof}
By induction on $n\ge 1$ together with the claim that for any $m>n$, $L^m\mwbsim^W_{2n+1} L^n$.
%\begin{enumerate}
%    \item $Z\not\mwbsim^W_{2n+2} L^n$
%    \item $L^m\not\mwbsim^W_{2n+2} L^n$
%    \item $L^n\mwbsim^W_{2n+1} L^m$
%    \item $L^n\mwbsim^W_{2n+1} Z$
%\end{enumerate}
Base case $L\mwbsim^W_3L^m$. Wlog.\ assume that $m$ minimizes $k$ in $L\mwbsim^W_kL^m$
and that \R\ only makes optimal moves
i.e. wins as quickly as possible. This means in particular that he needs to change the equivalence class
in every move. Thus, he can move either $L^m\wstep{a^*}L^{m'}Q^q\mwbsim L^{m'}$ or to $L^{m'}RQ^q$
for some $l<m$. In both cases \V\ to stays in $L$.
In the first case, because we assume optimal moves, we must have $L\not\mwbsim^W_iL^{l}$ for some $i<k$, which contradicts the optimality of $m$.
Alternatively, the game continues from $L^{l}RQ^q$ vs.\ $L$.
\R\ must again move from $L^{l}RQ^q$
and change the class. If he makes an $a$-step to $R$ or ends in $Q^i$ \V\ can match to the same process,
a move to some $RL^{l'}$ or $L^{l'}, l'<l\le m$ is surely non-optimal. The only remaining option is to move silently
to $R$ to which \V\ will respond by $L\wstep{}L$. Now observe that $L\mwbsim^W_1 R$.

Base case $Z\mwbsim^W_3 L$: As Z can silently go to $L^n$, \R\ needs to start from $Z$.
He has three options to change the class from here: to some $L^lQ^q\mwbsim L^l$,
to $RL^lQ^q\mwbsim RL^l$ or to something equivalent to $ZR$. In all cases \V\ responds
to $L$ and in the first two cases, we can use provious claims $L\mwbsim^W_3L^m$ and $L\mwbsim^W_2RL^m$ to conclude
that this allows him to survive another $2$ rounds. If the second round starts in 
$L$ vs.\ $ZR$ (or equivalent), \R\ can again not move from $L$ and has three options to change the class:
to something equivalent to $Z$ which is non-optimal as it repeats the initial configuration.
Alternatively he can go to $RL^lQ^1\mwbsim RL^l$ or to $RQ^q\mwbsim R$. In both cases we complete by
the observation that $RL^l\mwbsim^W_1 L\mwbsim^W_1 R$.

For the induction step, assume $L^m\mwbsim^W_{2n+1} L^n$ and $Z\mwbsim^W_{2n+1} L^n$. We show that
$L^m\mwbsim^W_{2(n+1)+1} L^{(n+1)}$:
Just as in the base case, the only good move for \R\ is
$L^m\wstep{a}L^{m'}RQ^q$ for some $n<m'<m$. \V\ in his response goes to $L^{n} R$.
Next one more time \R\ has the only one reasonable kind of move, to a process equivalent to $L^{m''}$, where $m''>n$.
However now \V\ responds to $L^n$ and we use the induction assumption to the pair $L^{m''}\mwbsim^W_{2n+1}L^n$.

Observe that because $\mwbsim^W_{2n+1}$ is a congruence this implies also
$L^mR\mwbsim^W_{2n+1}L^nR$ for $m\ge n$.
To show that $Z\mwbsim^W_{2(n+1)+1} L^{(n+1)}$ we assume wlog. that \R\ initially moves $Z\wstep{a}ZR$,
\V\ responds by $L^{n+1}\wstep{a}L^nR$. Now to prevent a perfect match in the next round, \R\ moves from $ZR$
to either $Z$ or to $L^mR$ or $L^m$. In the first case, \V\ will remove the $R$ and end up in $L^n$ and we
can use the induction assumption, in the last two cases \V\ stays in $L^nR$ or goes to $L^n$. Either way, we can use
the previous claims that $L^mR\mwbsim^W_{2n+1}L^nR$ and $L^m\mwbsim^W_{2n+1}L^n$ for $m\ge n$.
\qed
\end{proof}

\paragraph*{Claim (\ref{claim_3}): $Z\not\mwbsim^W_{2n+2} L^n$ for $n>0$.}
\begin{proof}
By induction on $n\ge 1$ together with the claim that for any $m>n$, $L^m\not\mwbsim^W_{2n+2} L^n$.
Base case: $Z\not\mwbsim^W_{4}L\not\mwbsim^W_{4}L^m$.
\R\ plays $L^m \wstep{a} LR$ (or $Z\wstep{a} LR$). Possible responses from $L$ are
\begin{enumerate}
    \item to $L Q^{q}$ or $Q^q$, from which \R\ wins in 3 rounds by Claim \ref{claim_1}.
    \item to $R Q^{q}$ in this case \R\ performs a move $L R \wstep{\tau} Q^{q+1}$ and \V\ responds to
        either $R Q^{i}$ or $Q^{i}$ with $i\leq q$. In both cases \R\ wins in one round
        by claim (\ref{claim_4}) or playing an $a^{q+1}$-step from $Q^{q+1}$.
\end{enumerate}
For the induction step we assume $L^m\not\mwbsim^W_{2n+2} L^n\not\mwbsim^W_{2n+2} Z$ and show that both
$L^m\not\mwbsim^W_{2(n+1)+2} L^{(n+1)}$ and $Z\not\mwbsim^W_{2(n+1)+2} L^{(n+1)}$ hold.
\R\ moves from $L^m$ (or $Z$) in an $a$-step to $L^{n+1}R$. \V\ can respond
\begin{enumerate}
    \item to $L^n Q^{q}$ or some $Q^q$, from which \R\ wins in 3 rounds by Claim \ref{claim_1}.
    \item to $L^{n'}R Q^{q}, n'<n$. In this case \R\ performs a move $L^n R \wstep{\tau} L^n$ and \V\ responds
        either to $L^{n''}R Q^{i}$ or $L^{n''}Q^{i}$ with $n''\leq n'<n$. In the first case, \R\ wins in one round
        by claim (\ref{claim_4}). In the last case the game continues from $L^n$ vs.\ $L^{n''<n}$ and we can use
        the induction assumption. \qed
\end{enumerate}

\end{proof}
    \begin{remark}
        To construct a counter-example to convergence of word approximants at level $\omega+k$ for finite $k$, the previous construction can
        be complemented by a "finite ladder", where $X$ and $Y$ are renamed to $X_0$ and $Y_0$: For $0<i\le k$ add
        variables $X_i,Y_i,Z_i,Z'_i, W_i, W'_i$ and rules as indicated below.
     
    \end{remark}
    
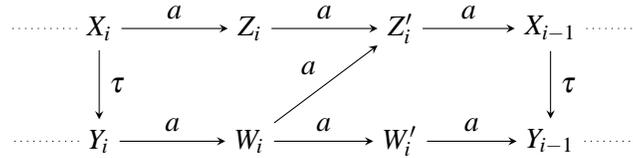
\begin{center}
          %\begin{minipage}[l]{.6\textwidth}
                \begin{tikzpicture}
                \node (X) {$X_i$};
                \node (Z) [right of=X] {$Z_i$};
		\node (Z1) [right of=Z] {$Z_i'$};
		\node (X1) [right of=Z1] {$X_{i-1}$};
                \node (Y) [below of=X,yshift=0.5cm] {$Y_i$}; 
                \node (W) [right of=Y] {$W_i$};
		\node (W1) [right of=W] {$W_i'$};
		\node (Y1) [right of=W1] {$Y_{i-1}$};
		\node (DSE) [right of=Y1,xshift=-0.7cm] {};
		\node (DSW) [right of=X1,xshift=-0.7cm] {};
		\node (DNE) [left of=Y,xshift=0.7cm] {};
		\node (DNW) [left of=X,xshift=0.7cm] {};
                    
                \path[->] (X) edge node {$a$} (Z);
		\path[->] (Z) edge node {$a$} (Z1);
		\path[->] (Z1) edge node {$a$} (X1);
		\path[->] (Y) edge node {$a$} (W);
		\path[->] (W) edge node {$a$} (W1);
		\path[->] (W1) edge node {$a$} (Y1);
  
		\path[->] (X) edge node {$\tau$} (Y);
		\path[->] (X1) edge node {$\tau$} (Y1);
		\path[->] (W) edge node {$a$} (Z1);
		
                \draw[dotted] (Y1) -- (DSE);
                \draw[dotted] (X1) -- (DSW);
                \draw[dotted] (Y) -- (DNE);
                \draw[dotted] (X) -- (DNW);
            \end{tikzpicture}
%           \end{minipage}
%            \begin{minipage}[r]{.35\textwidth}
%		\begin{tabular}{l}
%		    $X_i\step{a}Z_i$,\\
%                    $Z_i\step{a} Z'_i$,\\  
%                    $Z_i'\step{a} X_{i-1}$,\\
%                    \\
%                    $Y_i\step{a}W_i$,\\ 
%                    $W_i \step{a} W'_i$,\\
%                    $W'_i\step{a} Y_{i-1}$,\\
%                    \\
%                    $X_i\step{\tau}Y_i$,\\
%                    $W_i\step{a}Z'_i$.\\
%                \end{tabular}
%            \end{minipage}
        \end{center}

\section{Discussion}
In order to decide weak bisimulation for BPP or subclasses it suffices to provide a semi-decision
procedure for inequivalence. If we have some measure on which equivalent processes must agree,
we can define a new notion of approximants by additionally requiring that \V\ must preserve
equality on this measure in every round of an approximation game.
Conversely, one can think of properties as being captured by some notion of approximation $\mwbsim^O$:
If two processes disagree on the property then they are distinguished by $\mwbsim^O_i$ at some level $i\le\omega$.

As an example take the property \emph{norm preservation} of Claim 1) Proposition \ref{middle_lemma}:
Equivalent processes must have equal norms.
This is captured by Parikh or Word approximants because if two processes disagree on the norm, \R\ can distinguish
them in two rounds of the corresponding game by reducing the smaller one to a deadlock -- which cannot be done
in any proper response from the other side -- and playing an action from the non-deadlocked process
afterwards.
Another known invariant are the  \emph{distance to disabling} functions (dd-functions) used in \cite{Jan2003}
for strong bisimulation. If the shortest path from $\alpha$ to $\alpha'$
which disables any action $a$ is shorter than a shortest path from $\beta$ to a configuration which disables $a$
then $\alpha \not\mwbsim_2^P \beta$. So this first level of $dd$-functions is captured by Parikh approximants
at level $2$. We can continue this argument and say $n$-th order $dd$-functions
are capture by $\mwbsim_{n+1}^P$ relation. 

We have shown that all subclasses of BPP which are currently known to have decidable weak bisimulation are indeed
finitely approximable for some natural notion of approximation.
The lower bound of $\omega+\omega$ for the convergence of Word (and thus Parikh) approximants given by the construction in Theorem
\ref{thm:non_stabilize} leads us to the conclusion that we are in fact looking for a distinguishing property
that is orthogonal to Word approximants: It should still allow for decidable approximants
but at the same time it must be stronger than (not captured by) Word approximants
because otherwise it cannot be complete.

Our lower bound of $\omega*\omega$ for the symmetric short approximants $\mwbsim^L$ does not quite match the upper
bound of $\omega^\omega$ provided by \cite{HMS2006} and we conjecture that indeed, the exact convergence ordinal is $\omega * \omega$.

Finally, let us define the subclass of \emph{decreasing} systems in the following way.
\begin{definition}
 A $BPP$ description is \emph{decreasing} if there is a linear order on variables such that for every rule
 $X\lra[a] \alpha$ we have that $\alpha$ does not contain variables which are greater than $X$ in chosen order.
\end{definition}
It seems that this subclass provides much structure to work with. Nevertheless, all systems presented in this paper are
in fact decreasing. We believe that solving this class will be an important step towards a solution of the problem.

\paragraph{Acknowledgments.}
 We are grateful to the reviewers for their constructive feedback on earlier drafts of this document
 and thank Colin Stirling for many helpful discussions.
 
\bibliographystyle{eptcs}
\bibliography{eq}{}

% 
%\newpage
%\appendix
%\input{app_def}  % alternative definition for the approximants

%\input{app_pictures}

\end{document}